\RequirePackage{amsmath}
\documentclass[runningheads]{llncs}

\usepackage{microtype}
\usepackage{lmodern}
\usepackage{soul}
\usepackage{graphicx}
\usepackage[colorlinks,linkcolor=blue,citecolor=blue]{hyperref}
\usepackage{todonotes}
 
\usepackage{amssymb,amsfonts}
\usepackage{multirow,bigstrut,diagbox,booktabs}
\usepackage{subfigure}
\usepackage{enumitem}
\usepackage{stmaryrd} 
\usepackage{booktabs} 
\usepackage{color}
\usepackage[ruled,linesnumbered,vlined]{algorithm2e} 

\usepackage{natbib}
\usepackage{geometry}
\definecolor{emphcol}{gray}{.75}
\newcommand{\take}[1]{\colorbox{emphcol}{#1}}
\newcommand{\dmax}[1][\textbf{D}]{\|#1\|_{\infty}}

\begin{document}

\title{The curse of rationality in sequential scheduling games}

\author{Cong Chen\inst{1} \and Yinfeng Xu\inst{2}}

\institute{
School of Business Administration, South China University of Technology, Guangzhou, China \and
School of Management, Xi'an Jiaotong University, Xi'an, China
}
\maketitle

\begin{abstract}
	Despite the emphases on computability issues in research of algorithmic game theory,
	the limited computational capacity of players have received far less attention.
	This work examines how different levels of players' computational ability (or ``rationality'') impact the outcomes of \emph{sequential scheduling games}. Surprisingly, our results show that a lower level of rationality of players may lead to better equilibria.

	More specifically, we characterize the sequential price of anarchy (SPoA) under two different models of bounded rationality, namely, \emph{players with $k$-lookahead} and \emph{simple-minded players}. 
	The model in which players have $k$-lookahead interpolates between the ``perfect rationality'' ($k=n-1$) and ``online greedy'' ($k=0$).
	Our results show that the inefficiency of equilibria (SPoA) increases in $k$ the degree of lookahead: $\mathrm{SPoA} = O (k^2)$ for two machines and $\mathrm{SPoA} = O\left(2^k \min \{mk,n\}\right)$ for $m$ machines, where $n$ is the number of players.
	Moreover, when players are simple-minded, the SPoA is exactly $m$, which coincides with the performance of ``online greedy''.

\end{abstract}

\section{Introduction}\label{sec:intro}

Research on algorithmic game theory -- a fascinating fusion of both game theory and algorithms -- has attracted a lot of computer scientists and economists.
The core of this research field is to take the computability (computational complexity) into consideration while studying game theory problems, such as the complexity of finding Nash equilibria and the computational issues in mechanism design.
However, the computational ability of players has received little attention from the community, despite its strong ties to computational complexity and the actual behavior of players playing a game.
Most research assumes the players always have the ability to compute an optimal decision, even though sometimes finding an optimal decision is a very difficult problem (e.g. NPC problem).
Our work examines the impact of different levels of computational ability (also termed as ``rationality'' in this paper) of players on the outcomes of \emph{sequential scheduling games}. Surprisingly, the results show that a lower level of rationality may produce better equilibrium outcomes.

\paragraph{Sequential scheduling game (on unrelated machines).}
There are $n$ jobs $N=\left\{ 1,2,\ldots,n \right\}$ as players and $m$ machines $M=\left\{ 1,2,\ldots,m \right\}$ as strategies.
Each job $j$ will take $p_{i,j}$ units of time if processed by machine $i$.
The jobs sequentially choose one of the machines for processing, starting with job $1$ and ending with job $n$.
The load of a machine is the total processing times of the jobs processed on it.
The goal of each job is to choose a machine with a smallest possible load.

When a job makes decision, he knows the choices made by his predecessors as well as the processing times of his successors.
However, it is very hard for the job to compute an optimal decision.
Indeed, \citet{Leme2012} showed that for the unrelated machine scheduling computing a \emph{subgame-perfect equilibrium} (SPE) is PSPACE-complete.
One can glance at the example shown in Table~\ref{tab:example5jobs} to see how hard to find the optimal decision (in gray box) for job 1, and how easily the job may deviate from his optimal choice (to choose machine 2 with a very small processing time $\epsilon$) without enough computational ability.
\begin{table}[tb]
	\caption{An example from \citet{Giessler2016} with 5 jobs and 2 machines, where the SPE is shown as gray boxes.}
	\label{tab:example5jobs}
	\centering

	\begin{tabular}{c|c| c| c | c | c}
	\hline

	\hline
	& job 1 & job 2 & job 3 & job 4 & job 5 \\
	\hline
	machine 1 & \take{$3-11\epsilon$} & $\epsilon$ & \take{$\epsilon$} & $1-2\epsilon$ & $2-8\epsilon$ \\
	\hline
	machine 2 & $\epsilon$ & \take{$2-9\epsilon$} & $2-8\epsilon$ & \take{$1-2\epsilon$} & \take{$1-2\epsilon$} \\ 
	\hline

	\hline
	\end{tabular}
\end{table}

\paragraph{Price of anarchy and the curse of rationality.}

The concept of the price of anarchy (PoA), proposed by \citet{Koutsoupias2009} to assess the inefficiency of equilibria outcomes, has attracted many research over the past two decades.
To further understand the quality of SPEs outcomes of a game, \citet{Leme2012} introduced the \emph{sequential price of anarchy} (SPoA).
While the PoA compares the cost of a worst case Nash equilibrium to the optimal social cost,
the SPoA considers the outcomes of a sequential game where players, instead of choosing their strategies simultaneously, choose their strategies sequentially in some arbitrary order.

It turns out the PoA is very bad (unbounded) for even two unrelated machines, and introducing sequentiality only slightly improves the outcomes -- the SPoA grows linearly with the number $n$ of players \citep{Giessler2016}.
However, when we look carefully	at the worst case scenario, which gives the lower bound of SPoA for two unrelated machines in \citet{Giessler2016}, we find that the equilibrium is very unnatural and can hardly be achieved in reality, unless each player can solve a PSPACE-complete problem while making decision.
The example in Table~\ref{tab:example5jobs} already reveals the phenomenon that the first two players have to make a very complex computation to counter-intuitively choose a machine with a very high processing time rather than the one with almost 0 processing time.

Perhaps surprisingly, instead of assuming all the players have such strong rationality, if players are myopic (i.e., decisions are made only based on the predecessors' decision), the SPoA will be significantly improved to 2 for two unrelated machines, where the result can be deduced from the online greedy scheduling problem \citep{Aspnes1997}.
This result illustrates that full rationality may have a negative effect on the quality of outcomes.
Our work mainly investigates the impact of different levels of rationality on the SPoA for the unrelated machines scheduling game.

\paragraph{Modeling the bounded rationality.}
The notion of \emph{bounded rationality} can be traced back to the pioneering work of \citet{simon1955behavioral}.
Herbert Simon defines bounded rationality as ``rational choice that takes account the cognitive limitations of the decision-maker -- limitations of both knowledge and computational capacity''.
Frank Hahn remarks that ``there is only one way to be perfectly rational, while there are an infinity of ways to be partially rational...'' \citep{lee2011bounded}.
Indeed, there are tons of literature tried to model the bounded rationality.
We refer the readers to some surveys (see, e.g., \citet{velupillai2010foundations,lee2011bounded,di2016boundedly}) for details.

This paper propose two ways to model the bounded rationality of players:
\begin{enumerate}
	\item \emph{Players with $k$-lookahead.} We suppose each player only considers the next few successors' information for computing his decision, in addition to the known predecessors' decisions.
	We say a player has a \emph{$k$-lookahead} ability if he can compute the optimal decision depending on his next $k$ successors' information and the predecessors' decisions.
	Specifically, when a player makes decision, he will draw a $(k+1)$-level game tree (including the node of himself), assign the corresponding costs to the leaves, and then perform backward induction to decide which move to make.
	Similar setting can also be found in \citet{Mirrokni2012,Bilo2017,Groenland2018,Kroer2020}.

	\item \emph{Simple-minded players.} 
	As an extension, we also examine a situation where players make decisions only via simple calculations.
	When a so-called simple-minded player makes decision, he simply assumes the successors will choose machines with minimum processing times, so he can easily find a best choice depending on the assumption.
	The setting is also very natural in the unrelated machine scheduling, since choosing a machine with minimum processing time is mostly not a bad idea, and assuming other players doing so makes the prediction of other players' behaviors much simpler.
\end{enumerate}

\paragraph{Our contributions.}
This paper mainly investigates how the degree of rationality impacts the efficiency of SPEs.
We characterize the SPoA under two different models of bounded rationality, namely, \emph{players with $k$-lookahead} and \emph{simple-minded players}.
In general, quantifying the SPoA is a challenging task, and no general techniques are known in the literature.
In this paper, the key idea of most of our proofs is to characterize the amount of increase of the makespan or load of machine due to an additional job or set of jobs.
Out main results are as follows (see also Table~\ref{tab:results}):
 \begin{enumerate}
 	\item In Section~\ref{sec:1lookahead}, we first show that for sequential scheduling game on 2 unrelated machines the SPoA is $2$ for players with $1$-lookahead, which coincides with the case of $0$-lookahead -- i.e., online greedy.
 	This result perhaps suggests that the strategic behavior of only one player foreseen does not bring any negative influence on the current decision-maker. 
 	However, we will show in the following that the interaction of more than 2 players (i.e., $2$-lookahead) may have a negative effect on the decision being made. 
	\item For the players with $k$-lookahead, we obtain that $\mathrm{SPoA} = O(k^2)$ for 2 unrelated machines.
	This shows that the more lookahead the players have the worst the SPoA will be.
	But if we compare this result to the ``perfect rationality'' case where the SPoA is $\Theta (n)$ \citep{Giessler2016}, bounded rationality significantly improves the quality of SPEs. 
	(These results are presented in Section~\ref{sec:klookahead}.)
	\item We also characterize the SPoA for general $m$ unrelated machines case.
	We prove that $\mathrm{SPoA} = O\left(2^k \min \{mk,n\}\right)$ for players with $k$-lookahead,
	which also improves the $O(2^n)$ upper bound for the perfect rationality case. (See Section~\ref{sec:klookahead}.)
	\item At last, another bounded rationality model where the players are \emph{simple-minded} is discussed.
	It turns out that if assuming all the predecessors follow a simple rule -- choosing the machine with minimum processing time -- the player will make a decision as good as the online greedy, that is, $\mathrm{SPoA} = m$.
	(The results can be found in Section~\ref{sec:simple_minded_players}.)
 \end{enumerate}

\begin{table}[tb]
  \centering
  \caption{A summary of results, some of which achieved in this paper are marked with ``$*$''}
    \begin{tabular}{c|c|c}
    \hline

    \hline
          & $2$ machines & $m$ machines \\
    \hline
    Online greedy ($0$-lookahead) & $2$     & $m$ \\
    \hline
    $1$-lookahead & $2$ $^*$     & $O(m)$ $^*$ \\
    \hline
    $k$-lookahead & $O(k^2)$ $^*$ & $O\left(2^k \min \{mk,n\}\right)$ $^*$ \\
    \hline
    Perfect rationality ($n$-lookahead) & $\Theta (n)$ & $O(2^n)$ \\
    \hline
    Simple-minded & $2$ $^*$ & $m$ $^*$ \\
    \hline

    \hline
    \end{tabular}%
  \label{tab:results}%
\end{table}%

\paragraph{Further related work.}
The idea of \emph{limited lookahead} first appeared in the 1950s \citet{shannon1950programming}.
Recently, the idea has been investigated in some game-theoretic setting by several research \citep{Mirrokni2012,Kroer2020,Bilo2017,Groenland2018}.
In particular, \citet{Bilo2017} and \citet{Groenland2018} have very similar setting to our \emph{$k$-lookahead} model.
However, they both focused on the congestion games.
\citet{Bilo2017} studied the existence of $k$-lookahead equilibria and the PoA for $2$-lookahead (corresponding to $1$-lookahead in our setting) equilibria in congestion games with linear latencies.
\citet{Groenland2018} focused on the equilibria which are not only SPEs but also Nash equilibria.
They show that for generic simple congestion games the SPoA coincides with the PoA (independently of k).
In fact, both of the above work failed to reveal what the role of $k$ plays in the game, which is distinguished in the unrelated machines scheduling game by our work.

The SPoA for unrelated machines was first analyzed by \citet{Leme2012}, showing that $n \le \mathrm{SPoA} \le m\cdot 2^n$.
The bounds are improved to $2^{\Omega(\sqrt{n})} \le  \mathrm{SPoA} \le 2^n$ by \citet{Bilo2015}.
However, the above lower bounds use a non-constant number of machines, which means it is still unclear whether the lower bound is constant for constant number of machines.
\cite{Giessler2016} answered the question, showing that the SPoA is not constant for even two machines, that is, $\mathrm{SPoA} = \Omega (n)$.
They also provided a matching upper bound, concluding that $\mathrm{SPoA} = \Theta (n)$ for two unrelated machines.

\section{Preliminaries}\label{sec:notation}
We formally define the \emph{sequential scheduling game on unrelated machines} and our two models of bounded rationality -- \emph{players with $k$-lookahead} and \emph{simple-minded players}.

\paragraph{Sequential scheduling game on unrelated machines.}
Let $\left[ a:b \right] = \{a,a+1,\ldots,b\}$ and $\left[ b \right] = \left[ 1:b \right]$ where $a,b \in \mathbb{N}$.
Unrelated machine scheduling can be defined as a tuple $(N,M,(p_{i,j})_{i\in N, j\in M})$, 
where $N = [n]$ is the a set of jobs/players, $M=[m]$ is the set of machines/strategies, and $p_{i,j}$ is the processing time of job $j$ on machine $i$.
In sequential scheduling game, the jobs sequentially choose one of the machines for processing, starting with job $1$ and ending with job $n$.
A schedule $\sigma = (\sigma_1, \sigma_2, \ldots , \sigma_n)$ represents the decisions of the jobs, where $\sigma_j$ is the machine which job $j$ chooses.
The \emph{load} $L_i(N)$ of a machine $i$ in schedule $\sigma$ of jobs $N$ is the sum of the processing times of all jobs who choose machine $i$, that is, $L_i(N) = \sum_{j: \sigma_j=i}p_{i,j}$.
When job $j$ makes decision, he will try to minimize his own cost $L_{\sigma_j}(N)$ -- the load of machine he chooses -- taking into account all his predecessors and successors.
The schedule $\sigma$ is thus decided.
This is an \emph{extensive form game}, and so it always possesses \emph{subgame-perfect equilibria}, 
which can be calculated by backward induction.

Figure~\ref{fig:example3jobs} gives an example of 3 jobs, in which the ``perfect rationality'' part depicts the game tree for this example.
In this game, job 1 has to draw the whole tree, calculate the costs at each of the $2^3=8$ leaves, and find the best choice by backward induction.
The following jobs will also go through an associated subtree in a similarly fashion.
The bold lines show the subgame-perfect strategies, and the (unique) path from the root to the leaf corresponding to the black circle is the equilibrium solution (i.e., the schedule is $(2,1,1)$).

\begin{figure}[tb]
	\centering
	\includegraphics[width=1\textwidth]{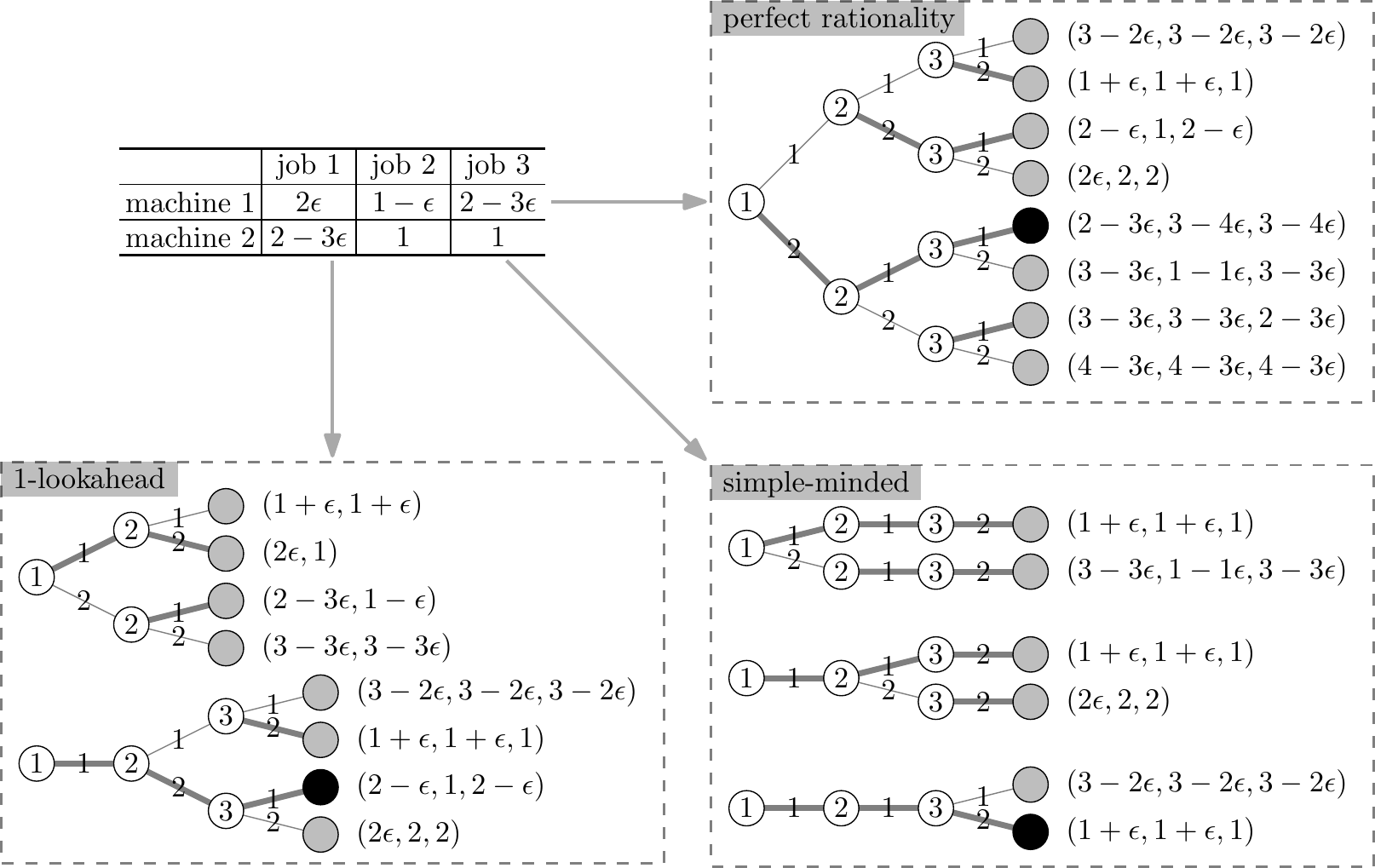}
	\caption{An example of 3 jobs and 2 machines.}
	\label{fig:example3jobs}
\end{figure}

\paragraph{Players with k-lookahead.}
In this model, we suppose each player can only foresee the next $k$ players.
Let $K_j = \left[ j+1:j+k \right]$ be the lookahead set for job $j$, where $|K_j| = k$.
When a player makes decision, he needs to draw a $(k+1)$-level game tree (consists of himself and the successors $K_j$), 
calculate the costs for this tree (with only $2^k$ leaves),
and find the best choose by backward induction.

The ``1-lookahead'' part in Figure~\ref{fig:example3jobs} gives an example of 3 jobs when players have 1-lookahead.
In this game, job 1 only knows the information of the next job (i.e., job 2), and considers the job as the last job.
After job 1 makes decision by backward induction, job 2 will compute a best choose according to job 1's decision and the next job's information.
As shown in the figure, the SPE for this example is $(1,2,1)$.

\paragraph{Simple-minded players.}
Simple-minded players ignore the strategical behaviors of their successors, and simply assume the successors will choose machines with minimum processing times.
According to this, the players can directly calculate the load of each machine and find a best one.

In Figure~\ref{fig:example3jobs}, the ``simple-minded'' part illustrates the decision process for each player.
It shows that each player just needs to calculate the costs for only 2 leaves, and selects one machine with a lower cost.
The resulting SPE for this example is $(1,1,2)$.

\paragraph{Inefficiency of equilibria.}
The (social) cost of a schedule/equilibrium is often defined as the \emph{makespan}, the maximum load over all machines.
To quantify the inefficiency of SPEs, \citet{Leme2012} introduced the sequential price of anarchy (SPoA) which compares the worst SPE with the optimal social cost:
\[
	\mathrm{SPoA} = \sup_J \frac{L_{\max}(J)}{OPT(J)},
\]
where $J$ takes over all possible job sets, and $OPT(J)$ is the makespan of the optimal schedule (a schedule generated by a central authority to minimize the social cost, and is not necessary an equilibrium) for jobs $J$.

For example, see the instance in Figure~\ref{fig:example3jobs}. 
The optimal schedule is $(1,1,2)$ with a makespan of $1+\epsilon$.
Since the SPE under perfect rationality is $(2,1,1)$ with a makespan of $3- 4 \epsilon$,
the SPoA for this example is $3$ (taking $\epsilon \rightarrow 0$).
However, when players have only 1-lookahead, the SPoA is $2$ (taking $\epsilon \rightarrow 0$).
Surprisingly, the SPE generated by simple-minded players is exactly the optimal schedule, that is, $\mathrm{SPoA} = 1$.

\paragraph{Additional notation.}
We introduce a notation of initial load $\mathbf{D} = (D_1,D_2, \ldots ,D_m)$ on the machines,
meaning the machines already have some initial load before the jobs playing a game.
Thus, $L_i(\mathbf{D},J)$ is the load of machine $i$ after the set $J$ of jobs play sequentially on the machines with a initial load $\mathbf{D}$,
and $L_{\max}(\mathbf{D},J) = \max_{i \in M} L_i(\mathbf{D},J)$ is the corresponding makespan. 
Notice that when $J = \emptyset$, $L_{\max}(\mathbf{D},\emptyset) = \dmax{}$.
Sometimes we use $\mathbf{D}(\ell{}) = (D_1(\ell{}),D_2(\ell{}), \ldots ,D_m(\ell{}))$ to represent the load of each machine due to the first $\ell{} \in N$ jobs,
where $D_i(\ell{}) = \sum_{j\in\{j|\sigma_j = i, 1 \le j \le \ell\}} p_{i,j}$ for $i=1,\ldots,m$.

To denote the maximum possible increase of the makespan due to the set $J$ of jobs for any initial load $\mathbf{D}\in \mathbb{R}^M_+$, we define:
\[
	\Delta L(J)=\sup_{\mathbf{D}\in \mathbb{R}^M_+}\left\{ L_{\max}(\mathbf{D},J)- \dmax{} \right\} \,.
\]
For each specific machine, we also define:
\[
	\Delta L_i(\mathbf{D},J) = L_i(\mathbf{D},J) - \dmax{} \,.
\]

For simplicity, we let $p_j = \min_{i \in M} p_{i,j}$ be the minimum processing time of job $j$,
and $x_{i,j}$ represent whether job $j$ chooses machine $i$ in the sequential game, i.e.,
\[
	x_{i,j}= 
	\begin{cases}
	1, &\text{ if job $j$ chooses machine $i$, that is, $\sigma_j = i$; }\\
	0, &\text{ otherwise. }
	\end{cases}
\]

\section{Players with 1-lookahead on two unrelated machines}\label{sec:1lookahead}
In this section, we analyze the SPoA for two unrelated machine when players has 1-lookahead.
We first prove a main lemma showing that the makespan is bounded by the total minimum processing time:
\begin{lemma}\label{lem.1look}
	$L_{\max}(\mathbf{0},N) \le \sum_{j\in N} p_j$.
\end{lemma}
\begin{proof}
	Since $L_{\max}(\mathbf{0},N) = \dmax[\textbf{D}(n)]{}$ by definition, we just prove that $\dmax[\textbf{D}(n)] \le \sum_{j\in N} p_j$.
	First, we will define a set of $\{n_0,n_1,\ldots,n_u\}$ where $n_\ell \in \{0,1,\ldots,n\}$ for $\ell = 0,1,\ldots,u$ and $n_u = n$.
	Then we prove a claim that $\dmax[\textbf{D}(n_\ell{})] \le \sum_{j=1}^{n_\ell{}}p_j$ for $\ell = 0,1,\ldots,u$, which indicates $\dmax[\textbf{D}(n)] \le \sum_{j\in N} p_j$ and proves this lemma.
	
	For a given set $N$ of jobs, their processing times $p_{i,j}$ and decisions $x_{i,j}$ in the sequential game,
	we define a set of $\{n_0,n_1,\ldots,n_u\}$ by Algorithm~\ref{alg:setn}. 

	We next show that the two ``while loops'' in Algorithm~\ref{alg:setn} (Line 5 and 8) will end at some $v \le n$.
	In other words, for example, if $D_1(v-1) + p_{1,v} \le D_2(v-1) + p_{2,v}$ (the first ``while loop''), there must be some $v'$ ($v \le v' \le n$) that $x_{1,v'} = 1$. 
	We take the first ``while loop'' as an example and the analysis for the second one is similar.
	When $D_1(v-1) + p_{1,v} \le D_2(v-1) + p_{2,v}$, machine 1 is a better choice for player $v$ regardless the decision of next player $v+1$.
	If player $v$ chooses machine 1 (i.e., $x_{1,v} = 1$), the loop ends.
	However, if player $v$ chooses machine 2 (i.e., $x_{1,v} = 0$), the only reason is that player $v$ knows player $v+1$ will choose machine 1 and $D_1(v-1) + p_{1,v} + p_{1,v+1} \ge D_2(v-1) + p_{2,v}$. 
	What makes layer $v$ believe player $v+1$ will choose machine 1 is that $D_1(v-1) + p_{1,v} + p_{1,v+1} \le D_2(v-1) + p_{2,v+1}$.
	Therefore, for player $v+1$ (after player $v$ has made his/her decision), machine 1 is a better choice regardless the decision of next player $v+2$.
	Similarly, if player $v+1$ chooses machine 2 (i.e., $x_{1,v+1} = 0$), it holds that player $v+1$ believes the next player $v+2$ will also choose machine 1.
	In a similar fashion, we know that if player $v''$ chooses machine 2 (i.e., $x_{1,v''} = 0$), player $v''+1$ will choose machine 1 regardless the decision of next player $v''+2$.
	Therefore, when player $v''+1$ is the last player (i.e. $v''+1 = n$), player $v''+1$ will surely choose machine 1 (i.e., $x_{1,v''+1} = 1$) and the loop ends.

	\begin{algorithm}[b]
		\KwIn{$p_{i,j}$ and $x_{i,j}$ for $i=1,2$ and $j = 1,\ldots,n$.}
		\KwOut{$\{n_0,n_1,\ldots,n_u\}$.}
		$u = 0$; $v = 1$; $D_1(0) = 0$; $D_2(0) = 0$; $n_0 = 0$\;
		$D_i(\ell{}) = \sum_{1 \le j \le \ell} p_{i,j} \cdot x_{i,j}$ for $\ell = 1,\ldots,n$ and $i=1,2$\;
		\While{$v \le n$}{
			\eIf({($n_u$ is the next $v$ that $x_{1,v} == 1$)}){$D_1(v-1) + p_{1,v} \le D_2(v-1) + p_{2,v}$}
			{
				\lWhile{$x_{1,v} == 0$}{$v$++}
				$u$++; $n_u=v$; $v$++\;
			}
			({($n_u$ is the next $v$ that $x_{2,v} == 1$)})
			{
				\lWhile{$x_{2,v} == 0$}{$v$++}
				$u$++; $n_u=v$; $v$++\;
			}
		}
		\caption{Definition of $\{n_0,n_1,\ldots,n_u\}$}
		\label{alg:setn}
	\end{algorithm}

	Given the set $\{n_0,n_1,\ldots,n_u\}$ where $n_0 = 0$ and $n_u = n$, we claim that:
	\begin{claim}\label{claim:1look}
		$\dmax[\textbf{D}(n_\ell{})] \le \sum_{j=1}^{n_\ell{}}p_j$ for $\ell = 0,1,\ldots,u$.
	\end{claim}
	We prove the claim by induction on $\ell = 0,1,\ldots,u$.
	The base case ($\ell = 0$) is trivial, since $\dmax[\textbf{D}(0)] = 0$ and $\sum_{j=1}^{0}p_j = 0$.
	Assume $\dmax[\textbf{D}(n_\ell{})] \le \sum_{j=1}^{n_\ell{}}p_j$ holds for $\ell = v$. We then prove that 
	\begin{equation}\label{eq:1look.induction}
		\dmax[\textbf{D}(n_{v'})] \le \sum_{j=1}^{n_{v'}}p_j \,,
	\end{equation}
	where $v'=v+1$.
	Since the first $n_v$ jobs create load $\textbf{D}(n_{v})$ on the machines (with a makespan $\dmax[\textbf{D}(n_{v})]$), 
	we only need to show that the increment of makespan after the allocation of job $\{n_v+1,n_v+2,\ldots,n_{v'}\}$ is no greater than the total minimum processing times of jobs $\{n_v+1,n_v+2,\ldots,n_{v'}\}$, namely,
	\[
		\dmax[\textbf{D}(n_{v'})] - \dmax[\textbf{D}(n_{v})] \le \sum_{j=n_v +1}^{n_{v'}}p_j \,.
	\]
	Therefore, we focus on the subgame played by players $\{n_v+1,n_v+2,\ldots,n_{v'}\}$.
	Without loss of generality, we consider the case ${D}_1(n_{v}) + p_{1,n_{v}+1} \le {D}_2(n_{v}) + p_{2,n_{v}+1}$ (the proof for the other case ${D}_1(n_{v}) + p_{1,n_{v}+1} > {D}_2(n_{v}) + p_{2,n_{v}+1}$ is similar).
	According to Algorithm~\ref{alg:setn}, jobs $\{n_v+1,n_v+2,\ldots,n_{v'}-1\}$ choose machine 2 and job $n_{v'}$ chooses machine 1 (as shown in the following Table~\ref{tab:1look.induction} where gray boxes represent the choices). 
	We know that
	\begin{equation}\label{eq.1look.claim.main}
		\dmax[\textbf{D}(n_{v'})] = \max \left\{ {D}_1(n_{v}) + p_{1,n_{v'}}\,,~{D}_2(n_{v}) + \sum_{j=n_v+1}^{n_{v'}-1}p_{2,j} \right\} \,.
	\end{equation}
	\begin{table}[h]
		\caption{Decisions of jobs $\{n_v+1,n_v+2,\ldots,n_{v'}\}$}
		\label{tab:1look.induction}
		\centering	
		\begin{tabular}{l|c|ccccc}
		\hline	

		\hline	
		machine 1 & ${D}_1(n_{v})$ & $p_{1,n_{v}+1}$ & $p_{1,n_{v}+2}$ & \ldots & $p_{1,n_{v'}-1}$  & \take{$p_{1,n_{v'}}$} \\
		\hline
		machine 2 & ${D}_2(n_{v})$ & \take{$p_{2,n_{v}+1}$} & \take{$p_{2,n_{v}+2}$} & \ldots & \take{$p_{2,n_{v'}-1}$}  & $p_{2,n_{v'}}$ \\
		\hline

		\hline
		\end{tabular}
	\end{table}

	In this case (${D}_1(n_{v}) + p_{1,n_{v}+1} \le {D}_2(n_{v}) + p_{2,n_{v}+1}$), we first give a upper bound on ${D}_1(n_{v}) + p_{1,n_{v}+1}$ prepared for the following proof.
	Since
	\[
		{D}_1(n_{v}) + p_{1,n_{v}+1} \le {D}_2(n_{v}) + p_{2,n_{v}+1} \le \max \left\{ {D}_1(n_{v}) ,~ {D}_2(n_{v}) \right\} + p_{2,n_{v}+1} 
	\]
	and
	\[
		{D}_1(n_{v}) + p_{1,n_{v}+1} \le \max \left\{ {D}_1(n_{v}) ,~ {D}_2(n_{v}) \right\} + p_{1,n_{v}+1} \,,
	\]
	we obtain that
	\begin{align}\label{eq:1look.claim.base}
		{D}_1(n_{v}) + p_{1,n_{v}+1} & \le \max \left\{ {D}_1(n_{v}) ,~ {D}_2(n_{v}) \right\} + \min \left\{p_{1,n_{v}+1},~ p_{2,n_{v}+1} \right\} \nonumber\\
		 & = \dmax[\textbf{D}(n_{v})] + p_{n_{v}+1} \,.
	\end{align}

	We then analyze properties that holds for the subgame.
	According to the decisions of players, we know that for any $n' \in \{n_v+1, n_v+2, \ldots , n_{v'}-1\}$, the reason why player $n'$ choose machine 2 is that he/she believes player $n'+1$ will choose machine 1 if he/she choose machine 1, i.e.,
	\begin{equation}\label{eq:1look.claim.eq1}
		{D}_1(n_{v}) + p_{1,n'} + p_{1,n'+1} \le {D}_2(n_{v}) + \sum_{j=n_v+1}^{n'-1}p_{2,j} + p_{2,n'+1}
	\end{equation}
	Thus, machine 2 is a better choice for player $n'$, i.e.,
	\begin{equation}\label{eq:1look.claim.eq2}
		{D}_2(n_{v}) + \sum_{j=n_v+1}^{n'}p_{2,j} \le {D}_1(n_{v}) + p_{1,n'} + p_{1,n'+1} \,.
	\end{equation}
	Define a function
	\[
		\Phi(n') = \max \left\{ {D}_1(n_{v}) + p_{1,n'+1}\,,~{D}_2(n_{v}) + \sum_{j=n_v+1}^{n'}p_{2,j} \right\} \,.
	\]
	Due to \eqref{eq:1look.claim.eq2}, we have
	\begin{align}
		\Phi(n') & \le \max\left\{ {D}_1(n_{v}) + p_{1,n'+1}\,,~{D}_1(n_{v}) + p_{1,n'} + p_{1,n'+1}  \right\} \nonumber\\
		\label{eq:1look.claim.eq3-1}
		&  = {D}_1(n_{v}) + p_{1,n'} + p_{1,n'+1} \\
		\label{eq:1look.claim.eq3-2}
		& \le \max \left\{ {D}_1(n_{v}) + p_{1,n'} \,,~ {D}_2(n_{v}) + \sum_{j=n_v+1}^{n'-1}p_{2,j} \right\} + p_{1,n'+1} \,.
	\end{align}
	Substituting \eqref{eq:1look.claim.eq1} into \eqref{eq:1look.claim.eq3-1} yields
	\begin{align}\label{eq:1look.claim.eq4}
		\Phi(n') & \le {D}_2(n_{v}) + \sum_{j=n_v+1}^{n'-1}p_{2,j} + p_{2,n'+1} \nonumber\\
		& \le \max \left\{ {D}_1(n_{v}) + p_{1,n'} \,,~ {D}_2(n_{v}) + \sum_{j=n_v+1}^{n'-1}p_{2,j} \right\}  + p_{2,n'+1} \,.
	\end{align}
	According to \eqref{eq:1look.claim.eq3-2} and \eqref{eq:1look.claim.eq4}, we get a critical inequality of $\Phi(n')$:
	\begin{align}\label{eq:1look.claim.eq5}
		\Phi(n') & \le \max \left\{ {D}_1(n_{v}) + p_{1,n'} \,,~ {D}_2(n_{v}) + \sum_{j=n_v+1}^{n'-1}p_{2,j} \right\}  + \min \left\{ p_{1,n'+1} ,~p_{2,n'+1} \right\} \\
		& = \Phi(n'-1) + p_{n'+1} \,. \nonumber
	\end{align}
	Therefore, we obtain that
	\[
		\Phi(n_{v'}-1) \le \Phi(n_{v'}-2) + p_{n_{v'}} \le \ldots \le \Phi(n_{v}+1) + \sum_{j=n_{v}+3}^{n_{v'}}p_{j} \,
	\]
	where
	\begin{align*}
		\Phi(n_{v}+1) & \le \max \left\{ {D}_1(n_{v}) + p_{1,n_{v}+1} \,,~ {D}_2(n_{v}) \right\}  + p_{n_{v}+2} & \text{by inequality \eqref{eq:1look.claim.eq5}} \\
		& \le \max \left\{ \dmax[\textbf{D}(n_{v})] + p_{n_{v}+1} \,,~ {D}_2(n_{v}) \right\}  + p_{n_{v}+2} & \text{by inequality \eqref{eq:1look.claim.base}} \\
		& = \dmax[\textbf{D}(n_{v})] + p_{n_{v}+1} + p_{n_{v}+2} \,.
	\end{align*}
	Therefore, it holds that
	\[
		\Phi(n_{v'}-1) \le \dmax[\textbf{D}(n_{v})] + \sum_{j=n_{v}+1}^{n_{v'}}p_{j} \,.
	\]
	Since \eqref{eq.1look.claim.main}, we have
	\[
		\dmax[\textbf{D}(n_{v'})] = \Phi(n_{v'}-1) \le \dmax[\textbf{D}(n_{v})] + \sum_{j=n_{v}+1}^{n_{v'}}p_{j} \,,
	\]
	which concludes the proof of Claim~\ref{claim:1look}, and therefore the lemma is proved.
	\qed
\end{proof}

\begin{theorem}
	For the sequential scheduling game on two unrelated machines where players have 1-lookahead, $\mathrm{SPoA} = 2$.
\end{theorem}
\begin{proof}
	According to Lemma~\ref{lem.1look} and an obvious lower bound on the optimal cost, namely $OPT(N) \ge {\sum_{j\in N} p_j}/{2}$, we obtain
	\[
		\mathrm{SPoA} = \frac{L_{\max}(\mathbf{0},N)}{OPT(N)} \le \frac{\sum_{j\in N} p_j}{{\sum_{j\in N} p_j}/{2}} = 2 \,.
	\]
	We then introduce a game that shows $\mathrm{SPoA} \ge 2$.
	There are only two jobs in this game (as shown in Table~\ref{tab:1look.lb}).
	The first job has processing times $1+\epsilon$ and $1$ on machines 1 and 2, respectively,
	and the second job has processing times $2$ and $1+\epsilon$ on machines 1 and 2, respectively.

	\begin{table}[h]
		\caption{A game of two players}
		\label{tab:1look.lb}
		\centering
	
		\begin{tabular}{c|c|c}
		\hline

		\hline
		 & job 1 & job 2 \\
		\hline
		machine 1 & $1+\epsilon$ & \take{$2$} \\
		\hline
		machine 2 & \take{$1$} & $1+\epsilon$ \\
		\hline

		\hline
		\end{tabular}
	\end{table}

	In this sequential game job 1 will choose machine 2, and thus job 2 will choose machine 1.
	The resulting makespan is 2.
	However, the optimal makespan is $1+\epsilon$, where job 1 chooses machine 1 and job 2 chooses machine 2.
	Therefore, we have $\mathrm{SPoA} \ge \frac{2}{1+\epsilon}$.
	By taking $\epsilon \to 0$, we obtain $\mathrm{SPoA} \ge 2$.
	\qed
\end{proof}

\section{Players with k-lookahead}\label{sec:klookahead} 
This section focuses on the general case where players have $k$-lookahead.
We first prove a key lemma showing that each job can only contribute a certain amount to the makespan:
\begin{lemma}\label{lem:klook.main}
	$\Delta L(\left[ \ell:n \right]) \le \Delta L(\left[ \ell+1:n \right]) + p_{\ell} + \Delta L(K_{\ell})$ for $\ell = 1,2,\ldots,n$.
\end{lemma}
\begin{proof}
	For $\ell \in \left[ 1:n \right]$, given a job set $\left[ \ell:n \right] $ and an initial load vector $\mathbf{D}$.
	We define two notations regarding the decision of job $\ell$.
	One is the new initial load after job $\ell$ chooses machine $i \in M$:
	\[
		\widetilde{\mathbf{D}}^{\shortrightarrow i} = \mathbf{D} + ( \underbrace{0,\ldots,0}_{i-1},~p_{i,\ell},~\underbrace{0,\ldots,0}_{m-i} ) \,.
	\]
	The other one is the anticipated cost of job $\ell$ with a lookahead set $K_{\ell}$ if he/she chooses machine $i \in M$:
	\[
		\widetilde{L}_i = L_i(\widetilde{\mathbf{D}}^{\shortrightarrow i},K_{\ell}) \,.
	\]

	Without loss of generality, we suppose job $\ell$ chooses machine $i^*$.
	Thus, the makespan for the game of the set $\left[ \ell:n \right]$ of players and the initial load $\mathbf{D}$ is
	\begin{align*}
		L_{\max}(\mathbf{D},\left[ \ell:n \right] ) & = L_{\max}(\widetilde{\mathbf{D}}^{\shortrightarrow i^*},\left[ \ell+1:n \right] ) \\
		& \le \dmax[\widetilde{\mathbf{D}}^{\shortrightarrow i^*}] + \Delta L (\left[ \ell+1:n \right]) \,.
	\end{align*}
	
	We first discuss a trivial case, where $\dmax[\widetilde{\mathbf{D}}^{\shortrightarrow i^*}] = \dmax[\mathbf{D}]$,
	that is, $\dmax[\mathbf{D}]$ will not increase after job $\ell$ chooses machine $i^*$.
	This case indicates that 
	\[
		L_{\max}(\mathbf{D},\left[ \ell:n \right] ) \le \dmax[\mathbf{D}] + \Delta L (\left[ \ell+1:n \right]) \,.
	\]
	The lemma is proved, since the inequality holds for any $\mathbf{D}$:
	\begin{align*}
		\Delta L(\left[ \ell:n \right]) & = \sup\left\{ L_{\max}(\mathbf{D},\left[ \ell:n \right])- \dmax{}: \mathbf{D}\in \mathbb{R}^M_+ \right\} \\
		& \le   \Delta L (\left[ \ell+1:n \right]) \,.
	\end{align*}

	Then we discuss the case $\dmax[\widetilde{\mathbf{D}}^{\shortrightarrow i^*}] > \dmax[\mathbf{D}]$.
	Because the increment of $\dmax[\mathbf{D}]$ is due to job $\ell$ chooses machine $i^*$,
	we know that $\dmax[\widetilde{\mathbf{D}}^{\shortrightarrow i^*}] = D_{i^*} + p_{i^*,\ell}$.
	This indicates that the anticipated cost of job $\ell$ is at least $\dmax[\widetilde{\mathbf{D}}^{\shortrightarrow i^*}]$, i.e.,
	\[
		\widetilde{L}_{i^*} \ge \dmax[\widetilde{\mathbf{D}}^{\shortrightarrow i^*}] \,.
	\]

	Let's focus on the moment when job $\ell$ makes decision.
	Job $\ell$ knows the initial load $\mathbf{D}$ and the lookahead set $K_{\ell}$.
	Thus, the anticipated cost of job $\ell$ for choosing any machine $i\in M$ is
	\[
		\widetilde{L}_{i} = L_{i}(\widetilde{\mathbf{D}}^{\shortrightarrow i},K_{\ell}) = \dmax[\widetilde{\mathbf{D}}^{\shortrightarrow i}] + \Delta L_{i} (\widetilde{\mathbf{D}}^{\shortrightarrow i},K_{\ell}) \,.
	\]
	Since job $\ell$ chooses machine $i^*$, it holds that
	\[
		\widetilde{L}_{i^*} \le \min_{i \in M} \left\{ \widetilde{L}_{i} \right\} = \min_{i \in M} \left\{ \dmax[\widetilde{\mathbf{D}}^{\shortrightarrow i}] + \Delta L_{i} (\widetilde{\mathbf{D}}^{\shortrightarrow i},K_{\ell}) \right\} \le \min_{i \in M} \left\{ \dmax[\widetilde{\mathbf{D}}^{\shortrightarrow i}] \right\} + \Delta L (K_{\ell}) \,.
	\]
	According to the definition of $\widetilde{\mathbf{D}}^{\shortrightarrow i}$, we know that 
	\[
		\dmax[\widetilde{\mathbf{D}}^{\shortrightarrow i}] \le \dmax[\mathbf{D}] + p_{i,\ell} \,.
	\]
	Thus it holds that
	\[
		\widetilde{L}_{i^*} \le \dmax[\mathbf{D}] + \min_{i \in M} \{p_{i,\ell}\} + \Delta L (K_{\ell}) = \dmax[\mathbf{D}] + p_{\ell} + \Delta L (K_{\ell}) \,.
	\]
	Since $\widetilde{L}_{i^*} \ge \dmax[\widetilde{\mathbf{D}}^{\shortrightarrow i^*}]$, it follows that
	\[
		\dmax[\widetilde{\mathbf{D}}^{\shortrightarrow i^*}] \le \dmax[\mathbf{D}] + p_{\ell} + \Delta L (K_{\ell})
	\]
	Therefore we have
	\begin{align*}
		L_{\max}(\mathbf{D},\left[ \ell:n \right] ) & \le \dmax[\widetilde{\mathbf{D}}^{\shortrightarrow i^*}] + \Delta L (\left[ \ell+1:n \right]) \\
		& \le  \dmax[\mathbf{D}] + p_{\ell} + \Delta L (K_{\ell}) + \Delta L (\left[ \ell+1:n \right]) \,.
	\end{align*}
	Since the inequality holds for any $\mathbf{D}$, we obtain
	\begin{align*}
		\Delta L(\left[ \ell:n \right]) & = \sup\left\{ L_{\max}(\mathbf{D},\left[ \ell:n \right])- \dmax{}: \mathbf{D}\in \mathbb{R}^M_+ \right\} \\
		& \le  p_{\ell} + \Delta L (K_{\ell}) + \Delta L (\left[ \ell+1:n \right]) \,,
	\end{align*}
	which completes the proof.
	\qed
\end{proof}

\begin{theorem}\label{thm:klookahead.2&mmachines}
	For the sequential scheduling game where players have $k$-lookahead, the SPoA is at most $O(k^2)$ for the two unrelated machines case, and at most $O(2^k \cdot \min\{mk,n\})$ for the $m$ unrelated machines case.
\end{theorem}
\begin{proof}
	According to Lemma~\ref{lem:klook.main}, it follows that
	\begin{equation}\label{eq:SPoA.klook.mmachines.deltaL}
		\Delta L([1:n]) \le \Delta L([2:n]) +p_1 + \Delta L (K_1) \le \ldots \le 
		\sum_{j=1}^n p_j + \sum_{j=1}^n \Delta L (K_j) \,.
	\end{equation}
	For the two unrelated machines case, we can know from \citet[Theorem~4]{Giessler2016} that
	\[
		\Delta L (K_j) \le (k-1) \sum_{j\in K_j} p_j \,.
	\]
	So we have
	\begin{align*}
		\Delta L([1:n]) & \le \sum_{j=1}^n p_j + \sum_{j=1}^n \Delta L (K_j) \\
		& \le \sum_{j=1}^n p_j +  (k-1) \left (\sum_{j \in K_1} p_j + \sum_{j \in K_2} p_j + \ldots + \sum_{j \in K_n} p_j \right) \\
		& \le \sum_{j=1}^n p_j +  k(k-1) \sum_{j=1}^n p_j \\ 
		& \le (k^2 - k +1) \sum_{j=1}^n p_j \,.
	\end{align*}
	Therefore, it holds that
	\[
		\mathrm{SPoA} \le \frac{\Delta L (N)}{OPT(N)} \le \frac{(k^2 - k +1) \sum_{j=1}^n p_j}{\sum_{j=1}^n p_j/2} = 2 (k^2 - k +1) = O (k^2) \,,
	\]
	for 2 unrelated machines.

	As for the $m$ unrelated machines case, according to \citet[Theorem 4]{Leme2012} and \citet[Theorem 5]{Bilo2015}, we get two upper bounds for $\Delta L (K_j)$:
	\begin{gather}
		\label{eq:SPoA.klook.mmachines.sump}
		\Delta L (K_j) \le 2^k \sum_{j\in K_j} p_j \,,\\
		\label{eq:SPoA.klook.mmachines.maxp}
		\Delta L (K_j) \le 2^k \max_{j\in K_j} p_j \,.
	\end{gather}
	On one hand, by \eqref{eq:SPoA.klook.mmachines.deltaL}, \eqref{eq:SPoA.klook.mmachines.sump} and $OPT(N) \ge \sum_{j\in N} p_j /m$, we have
	\begin{align*}
		\mathrm{SPoA} & \le \frac{\Delta L (N)}{OPT(N)} \\
		& \le \frac{\sum_{j\in N} p_j + \sum_{j\in N} \Delta L (K_j)}{\sum_{j\in N} p_j /m} \\
		& \le \frac{\sum_{j\in N} p_j + k 2^k \sum_{j\in N} p_j}{\sum_{j\in N} p_j /m} \\
		& = m k 2^k + m \\
		& = O(mk2^k)  \,.
	\end{align*}
	On the other hand, by \eqref{eq:SPoA.klook.mmachines.deltaL}, \eqref{eq:SPoA.klook.mmachines.maxp} and $OPT(N) \ge \max_{j\in N} p_j$, we have
	\begin{align*}
		\mathrm{SPoA} & \le \frac{\Delta L (N)}{OPT(N)} \\
		& \le \frac{\sum_{j\in N} p_j + \sum_{j\in N} \Delta L (K_j)}{OPT(N)} \\
		& \le \frac{\sum_{j\in N} p_j}{OPT(N)} + \frac{n 2^k \max_{j\in N} p_j}{OPT(N)} \\
		& \le \frac{\sum_{j\in N} p_j}{\sum_{j\in N} p_j /m} + \frac{n 2^k \max_{j\in N} p_j}{\max_{j\in N} p_j} \\
		& = m + n 2^k \\
		& = O(n 2^k)  \,.
	\end{align*}
	Therefore, we obtain $\mathrm{SPoA} = O(2^k \cdot \min\{mk,n\})$ for $m$ unrelated machines.
	\qed
\end{proof}

\section{Simple-minded players} 
\label{sec:simple_minded_players}

A simple-minded player makes decision only via simple calculations.
When a simple-minded player $j$ makes decision, job $j$ will select a machine with minimum anticipated load assuming that all the follow-up players will simply choose machines with minimum processing time. 
We show in this section that the SPoA is exactly $m$, the number of machines.

\paragraph*{Additional notation.}
For any job $j$ and machine $i$, we define
\[
	A_i(j) = D_i(j) + P_i([j+1:n]) \,,
\]
where $D_i(j)$ is the initial load of machine $i$ due to the first $j$ jobs, $P_i([j+1:n])$ is the total processing time of the jobs who are assumed by job $j$ to choose machine $i$ (i.e. the jobs have minimum processing time on machine $i$).
Note that the definition of $P_i(\cdot)$ avoids the issue caused by a tie, since a job can assume the follow-up players to choose only one machine.

\begin{theorem}
	For the sequential scheduling game on $m$ unrelated machines where players are \emph{simple-minded}, $\mathrm{SPoA} \le m$.
\end{theorem}

\begin{proof}
	We claim that $A_{\max} (\ell) \le A_{\max} (\ell-1)$ for $\ell = 1,2,\ldots,n$, where $A_{\max} (j) = \max_{i\in M} A_i(j)$.
	If the claim is true, we can obtain $A_{\max} (n) \le A_{\max} (0)$.
	Since $A_{\max} (n) = L_{\max}(\mathbf{0},N)$, $\sum_{j\in N} A_j (0) = \sum_{j\in N} p_j$ (by definition) and $OPT(N) \ge \sum_{j\in N} p_j / m$, it holds that
	\[
	 	\mathrm{SPoA} 	= 	\frac{L_{\max}(\mathbf{0},N)}{OPT(N)} 
					 	= 	\frac{A_{\max} (n)}{OPT(N)} 
					 	\le \frac{A_{\max} (0)}{OPT(N)}
					 	\le \frac{\sum_{j\in N} A_j (0)}{OPT(N)}
					 	\le \frac{\sum_{j\in N} p_j}{\sum_{j\in N} p_j / m} 
					 	= 	m \,,
	\]
	meaning the theorem holds.

	We then prove the above claim.
	For an arbitrary $\ell$, we assume machine $i^*$ has the minimum processing time for job $\ell$, i.e., $p_{\ell} = p_{i^*,\ell}$.

	\paragraph{Case 1.} If job $\ell$ chooses machine $i^*$, it follows that
	\begin{align*}
		A_{i^*}(\ell)	& = D_{i^*}(\ell) + P_{i^*}([\ell+1:n]) \\
						& = D_{i^*}(\ell - 1) + p_{\ell} + P_{i^*}([\ell+1:n]) \\
						& = D_{i^*}(\ell - 1) + P_{i^*}([\ell:n]) \\
						& = A_{i^*}(\ell - 1) \,.
	\end{align*}
	Since job $\ell$ chooses machine $i^*$, for any other machine $i \in M \backslash\{i^*\}$, the decision of job $\ell$ does not change $A_{i}(\ell - 1)$, that is, $A_{i}(\ell) = A_{i}(\ell - 1)$.
	Therefore we have $A_{\max}(\ell) = A_{\max}(\ell - 1)$.

	\paragraph{Case 2.} If job $\ell$ chooses machine $i' \neq i^*$, this means machine $i'$ is a better choice than machine $i^*$ for job $\ell$, i.e., $A_{i'}(\ell) + p_{i',\ell} \le A_{i^*}(\ell)$.

	For machine $i^*$, it holds that
	\begin{align*}
		A_{i^*}(\ell)	& = D_{i^*}(\ell) + P_{i^*}([\ell+1:n]) \\
						& = D_{i^*}(\ell - 1) + P_{i^*}([\ell+1:n]) \\
						& = D_{i^*}(\ell - 1) + P_{i^*}([\ell:n]) - p_{\ell} \\
						& = A_{i^*}(\ell - 1) - p_{\ell} \,.
	\end{align*}
	Thus, we obtain that $A_{i'}(\ell) \le A_{i^*}(\ell - 1)$ and $A_{i^*}(\ell) \le A_{i^*}(\ell - 1)$.
	Since for any other machine $i \in M \backslash\{i^*,i'\}$, the decision of job $\ell$ does not change $A_{i}(\ell - 1)$ (i.e., $A_{i}(\ell) = A_{i}(\ell - 1)$), we have $A_{\max}(\ell) \le A_{\max}(\ell - 1)$.

	The above two cases show that the claim is true, completing the proof.
	\qed
\end{proof}

\begin{theorem}
	For the sequential scheduling game on $m$ unrelated machines where players are \emph{simple-minded}, $\mathrm{SPoA} \ge m$.
\end{theorem}
\begin{proof}
	We first give an example in Table~\ref{tab:instance_for_4_machines} to illustrate a game in which $\mathrm{SPoA} = m$ for $m=4$.
	To avoid breaking ties, we add $\epsilon$ terms so that ties never occur.
	In the example, gray boxes represent the choices of jobs and the optimal solution is given in \textbf{bold}.
	More specifically, job 1 thinks the follow-up jobs will choose machines with minimum processing time (i.e. jobs 2-3 will choose machine 1 with processing time $1- \epsilon$).
	Thus machine 2 is the best choice for job 1, since the anticipated cost $4-4 \epsilon$ of choosing machine 1 is higher than $4-5 \epsilon$ of choosing machine 2.
	Likewise, jobs 2 will choose machine 3, job 3 will choose machine 4, and the last job 4 will choose machine 1.
	The optimal solution is obviously that each job $i$ chooses machine $i$.
	Therefore, we can obtain $\mathrm{SPoA} = {4}$ by taking $\epsilon \rightarrow 0$.

	\begin{table}[tb]
	  \centering
	  \caption{An instance for 4 machines}
	    \begin{tabular}{c|c|c|c|c}
	    \hline

	    \hline
	          & job 1     & job 2     & job 3     & job 4 \\
	    \hline
	    machine 1     & $\mathbf{1-\epsilon}$ & $1-\epsilon$ & $1-\epsilon$ & \take{$1-\epsilon$} \\
	    \hline
	    machine 2     & \take{$4-5\epsilon$} & $\mathbf{1}$     & $\infty$     & $\infty$ \\
	    \hline
	    machine 3     & $\infty$     & \take{$3-4\epsilon$} & $\mathbf{1}$     & $\infty$ \\
	    \hline
	    machine 4     & $\infty$     & $\infty$     & \take{$2-3\epsilon$} & $\mathbf{1}$ \\
	    \hline

	    \hline
	    \end{tabular}%
	  \label{tab:instance_for_4_machines}%
	\end{table}%

	For the general $m$ machines case, we can provide an instance of $m$ jobs in a similar fashion:
	\begin{enumerate}
		\item the processing time $p_{1,j}$ for each $j \in \left\{ 1,2,\ldots,m \right\}$ on machine 1 is $1- \epsilon$ (also the minimum processing time);
		\item the processing time $p_{j,j}$ for job $j \in \left\{ 2,3,\ldots,m \right\}$ on machine $j$ is $1$;
		\item the processing time $p_{j+1,j}$ for job $j \in \left\{ 1,2,\ldots,m-1 \right\}$ on machine $j+1$ is $m+1-j-(m+2-j)\epsilon$;
		\item all the undefined processing times are $\infty$.
	\end{enumerate}

	In this instance, the decisions of the jobs are: job $j$ chooses machine $j+1$ for $j = 1,2,\ldots,m-1$, and job $m$ chooses machine 1, resulting a makespan of $m-(m+1)\epsilon$.
	However, the optimal makespan is $1$ achieved by each job $j$ chooses machine $j$.
	By taking $\epsilon \rightarrow 0$, we can obtain $\mathrm{SPoA} = {m}$.
	Since we can find an instance that $\mathrm{SPoA} = {m}$, it follows that the SPoA is at least $m$.
	\qed
\end{proof}

\section{Conclusion} 
\label{sec:conclusion}

One of our main contributions to the area of algorithmic game theory is the reconsideration of ``perfect rationality'' assumption for the players.
This work helps, in some degree, to understand why some games in reality perform much better than the theoretical prediction.
As an example, the inefficiency of the subgame-perfect equilibrium for scheduling game on two unrelated machines is unbounded (for unbounded number of players).
However it might not be so bad in realty, since the real world players might only have bounded rationality.
Our results just explain this phenomenon in a theoretical way.
We believe this work takes a promising step in further understanding the role that bounded rationality plays in algorithmic game theory.

 \bibliographystyle{apalike}
 \bibliography{mybib}

\end{document}